\documentclass{article}%
\usepackage{amsfonts}
\usepackage{amsmath}
\usepackage{amssymb}
\usepackage{graphicx}%
\providecommand{\U}[1]{\protect\rule{.1in}{.1in}}
\newtheorem{theorem}{Theorem}
\newtheorem{acknowledgement}[theorem]{Acknowledgement}

\newtheorem{lemma}[theorem]{Lemma}

\newtheorem{proposition}[theorem]{Proposition}

\newenvironment{proof}[1][Proof]{\noindent\textbf{#1.} }{\ \rule{0.5em}{0.5em}}
\begin{document}

\title{Bypassing the Groenewold--van Hove Obstruction on $\mathbb{R}^{2n}$: A New
Argument in Favor of Born--Jordan Quantization}
\author{Maurice A. de Gosson\thanks{maurice.de.gosson@univie.ac.at}\\University of Vienna\\Faculty of Mathematics, NuHAG\\Oskar-Morgenstern-Platz 1, 1090 Vienna}
\maketitle

\begin{abstract}
There are known obstructions to a full quantization of $\mathbb{R}^{2n}$ in
the spirit of Dirac's approach, the most known being the Groenewold--van Hove
no-go result. We show, following a suggestion of S. K. Kauffmann, that it is
possible to construct a well-defined quantization procedure by weakening the
usual requirement that commutators should correspond to Poisson brackets. The
weaker requirement consists in demanding that this correspondence should only
hold for Hamiltonian functions of the type $T(p)+V(q)$. This reformulation
leads to a non-injective quantization of all observables $H\in\mathcal{S}%
^{\prime}(\mathbb{R}^{2n})$ which, when restricted to polynomials, is the rule
proposed by Born and Jordan in the early days of quantum mechanics.

\end{abstract}

\section{Introduction}

PACS:\ 0-01, 0-03

The problem of quantization harks back to the early years of quantum mechanics
when physicists were confronted to ordering problems (see, in this context,
the well-documented reviews by Ali and Englis \cite{tareq} and Castellani
\cite{20}). In the present paper we will deal more specifically with what is
sometimes called \textquotedblleft canonical quantization\textquotedblright,
which is a procedure for finding the quantum analogue of a classical theory
(in this case Hamiltonian mechanics), while attempting to preserve the formal
structure, such as symmetries, of the classical theory, to the greatest extent
possible. By definition a canonical quantization of $\mathcal{S}^{\prime
}(\mathbb{R}^{2n})$ is a continuous linear map
\[
\operatorname*{Op}:\mathcal{S}^{\prime}(\mathbb{R}^{2n})\longrightarrow
\mathcal{L(S}(\mathbb{R}^{n}),\mathcal{S}^{\prime}(\mathbb{R}^{n}))
\]
associating to each $H\in\mathcal{S}^{\prime}(\mathbb{R}^{2n})$ a continuous
linear operator
\[
\widehat{H}=\operatorname*{Op}(H):\mathcal{S}(\mathbb{R}^{n})\longrightarrow
\mathcal{S}^{\prime}(\mathbb{R}^{n})
\]
and such that some additional properties hold, for instance:

\begin{description}
\item[(CQ1)] $\operatorname*{Op}(q_{j}^{r})\psi(q)=q_{j}^{r}\psi(q)$
\textit{and} $\operatorname*{Op}(p_{j}^{s})\psi(q)=(-i\hbar\partial_{q_{j}%
})^{s}\psi(q)$ \textit{for all} $\psi\in\mathcal{S}(\mathbb{R}^{n})$
(\textit{hence in particular} $\operatorname*{Op}(1)=I_{\mathrm{d}}$);

\item[(CQ2)] (\textit{Born's CCR}) $[\widehat{q}_{j},\widehat{p}_{k}%
]=i\hbar\delta_{jk}$ for all $j,k$;

\item[(CQ3)] \textit{When }$H\in\mathcal{S}^{\prime}(\mathbb{R}^{2n})$
\textit{is real then} $\widehat{H}$ \textit{is a symmetric operator defined on
the dense subspace} $\mathcal{S}(\mathbb{R}^{n})$ of $L^{2}(\mathbb{R}^{n}%
)$\textit{.}
\end{description}

We notice that the usual Weyl quantization satisfies the axioms above; it is
also the preferred quantization in physics since it is in a sense the one
which sticks the closest to the classical structures (it allows to preserve
the covariance of Hamiltonian mechanics under linear canonical
transformations; for detailed accounts see \cite{Birk,TRANS,75}). In addition
to the axioms above, one has tried, following Dirac's program \cite{Dirac}, to
require that commutators correspond (up to a constant) to Poisson brackets:
\begin{equation}
\lbrack\operatorname*{Op}(H),\operatorname*{Op}(K)]=i\hbar\operatorname*{Op}%
(\{H,K\}). \label{Dirac}%
\end{equation}
The rub comes from the seminal papers by Groenewold \cite{groenewold} and van
Hove \cite{vanhove1,vanhove2}, who showed that this \textquotedblleft
commutator $\longleftrightarrow$ Poisson bracket\textquotedblright%
\ correspondence cannot hold for all observables (see \cite{berndt,gotuy,GS}
for detailed discussions). In fact, one shows, after some preparatory work
involving Poisson algebras of polynomials that one is led to a contradiction.
In fact, assuming $n=1$ one sets out to quantize $q^{2}p^{2}$. Using the rules
$\operatorname*{Op}(q^{r})=(\operatorname*{Op}(q))^{r}$ and
$\operatorname*{Op}(p^{r})=(\operatorname*{Op}(p))^{r}$ the application of
(\ref{Dirac}) to the trivial identity%
\[
q^{2}p^{2}=\frac{1}{9}\{q^{3},p^{3}\}=\frac{1}{3}\{q^{2}p,p^{2}q\}
\]
leads to the conflicting formulas%
\begin{equation}
\operatorname*{Op}(q^{2}p^{2})=\frac{1}{9}\operatorname*{Op}\{q^{3}%
,p^{3}\}=(\widehat{q})^{2}(\widehat{p})^{2}-2i\hbar\widehat{q}\widehat
{p}-\frac{2}{3}\hbar^{2} \label{gh1}%
\end{equation}
and
\begin{equation}
\operatorname*{Op}(q^{2}p^{2})=\frac{1}{3}\operatorname*{Op}\{q^{2}%
p,p^{2}q\}=(\widehat{q})^{2}(\widehat{p})^{2}-2i\hbar\widehat{q}\widehat
{p}-\frac{1}{3}\hbar^{2}; \label{gh2}%
\end{equation}
one concludes that there is thus no quantization satisfying Dirac's
correspondence for all monomials.

In the present work we show that these difficulties can be overcome (in a
physically satisfactory way) if one relaxes the general Dirac correspondence
and replaces it with a weaker condition, suggested by Kauffmann
\cite{Kauffmann}, namely that (\ref{Dirac}) only holds for Hamiltonian
functions which are of the type \textquotedblleft generalized kinetic energy
plus potential\textquotedblright\ $T(p)+V(q)$. We will see that this weaker
assumption allows to construct a quantization procedure for all tempered
distributions on $\mathbb{R}^{2n}$ which, when restricted to monomials
$q^{r}p^{s}$, is that proposed by Born and Jordan \cite{16} and which we have
extensively studied \cite{cogoni1,TRANS,physrep,Springer}. This result is thus
another argument in favor of Born--Jordan quantization. (We notice that the
idea of by-passing the Groenewold--van Hove obstruction by some means is not
quite new, see Gotay's paper \cite{gotay}).

Using the notation $X=\{0\}\times\mathbb{R}^{n}$ and $X^{\ast}=\mathbb{R}%
^{n}\times\{0\}$ we will show that

\begin{description}
\item[(BJQ1)] $[\operatorname*{Op}(T),\operatorname*{Op}(V)]=i\hbar
\operatorname*{Op}(\{T,V\})$\textit{ for all} $T\in C^{\infty}(X^{\ast})$ and
$V\in C^{\infty}(X)$ that are $\mathcal{S}^{\prime}(\mathbb{R}^{n})$;
\end{description}

The axiom (BJQ1) will be referred to as the \emph{reduced Dirac condition};
using the linearity of the Poisson bracket, it is equivalent to the axiom:

\begin{description}
\item[(BJQ1bis)] $[\widehat{H},\widehat{K}]=i\hbar\operatorname*{Op}(\{H,K\})$
\textit{for all} $H,K\in C^{\infty}(X)\oplus C^{\infty}(X^{\ast})$
\textit{that are in }$\mathcal{S}^{\prime}(\mathbb{R}^{2n})$\textit{.}
\end{description}

where $C^{\infty}(X)\oplus C^{\infty}(X^{\ast})$ is the space of all functions
$V(x)+T(p)$.

The rule (BJQ1bis) in particular applies to all Hamiltonians of the physical
type \textquotedblleft kinetic energy + potential\textquotedblright. Notice
that (BJQ1) implies that for all integers $r,s>0$ we have
\begin{equation}
\operatorname*{Op}(\{q_{j}^{r},p_{j}^{s}\})=i\hbar rs\operatorname*{Op}%
(q_{j}^{r-1}p_{j}^{s-1}) \label{ihrs}%
\end{equation}
and $\operatorname*{Op}(\{q_{j}^{r},p_{k}^{s}\})=0$ if $j\neq k$.

\section{Quantization of Monomials}

We will use the following commutation relations valid for all operators
$\widehat{q}$ and $\widehat{p}$ satisfying the CCR $[\widehat{q},\widehat
{p}]=i\hbar$:%

\begin{equation}
\lbrack(\widehat{q})^{r},(\widehat{p})^{s}]=si\hbar\sum\limits_{j=0}%
^{r-1}(\widehat{q})^{r-1-j}(\widehat{p})^{s-1}(\widehat{q})^{j}=ri\hbar
\sum\limits_{j=0}^{s-1}(\widehat{p})^{s-1-j}(\widehat{q})^{r-1}(\widehat
{p})^{j} \label{commut1}%
\end{equation}
(we will give a proof of this equality in the Appendix).

\begin{lemma}
\label{lem1}Let $r\geq0$ be an integer. We have
\begin{gather}
\widehat{q_{j}^{r}}=(\widehat{q_{j}})^{r}\text{ \ , \ }\widehat{p_{j}^{r}%
}=(\widehat{p_{j}})^{r}\label{vn1}\\
\widehat{q_{j}p_{j}}=\frac{1}{2}(\widehat{q}_{j}\widehat{p}_{j}+\widehat
{p}_{j}\widehat{q}_{j}). \label{anti1}%
\end{gather}

\end{lemma}

\begin{proof}
It is sufficient to assume $n=1$ and $r>0$. We have
\[
\lbrack(\widehat{q})^{r+1},\widehat{p}]=i\hbar\operatorname*{Op}%
(\{q^{r+1},p\})=i\hbar(r+1)(\widehat{q})^{r}%
\]
hence, using the second equality (\ref{commut1}),
\[
(\widehat{q})^{r}=\frac{1}{i\hbar(r+1)}[(\widehat{q})^{r+1},\widehat
{p}]=(\widehat{q})^{r}.
\]
The formula $\widehat{p^{r}}=(\widehat{p})^{r}$ is proven by a similar
argument, writing $[\widehat{q},(\widehat{p})^{r+1}]=i\hbar\operatorname*{Op}%
(\{q,p^{r+1}\})$. To prove (\ref{anti1}) it suffices to note that, since
$\{q_{j}^{2},p_{j}^{2}\}=4q_{j},p_{j}$ we have, using the commutation formula
(\ref{commut1}),
\[
\operatorname*{Op}(qp)=\frac{1}{4i\hbar}[(\widehat{q})^{2},(\widehat{p}%
)^{2}]=\frac{1}{2}(\widehat{q}\widehat{p}+\widehat{p}\widehat{q}).
\]

\end{proof}

Let us now show that formula (\ref{ihrs}) allows, as claimed in the
introduction, an unambiguous quantization of monomials in the $q_{j},p_{k}$
variables. We recall \cite{Springer,34,JCPain} that the Born--Jordan
quantization of a monomial $q_{j}^{r}p_{j}^{s}$ is given by the equivalent
formulas
\begin{align}
\operatorname*{Op}\nolimits_{\mathrm{BJ}}(q_{j}^{r}p_{j}^{s})  &  =\frac
{1}{r+1}\sum_{\ell=0}^{r}(\widehat{q_{j}})^{r-\ell}(\widehat{p_{j}}%
)^{s}(\widehat{q_{j}})^{\ell}\label{BJ1}\\
\operatorname*{Op}\nolimits_{\mathrm{BJ}}(q_{j}^{r}p_{j}^{s})  &  =\frac
{1}{s+1}\sum_{\ell=0}^{s}(\widehat{p_{j}})^{r-\ell}(\widehat{q_{j}}%
)^{s}(\widehat{p_{j}})^{\ell}. \label{BJ2}%
\end{align}

\begin{proposition}
\label{prop1}We have for all integers $r,s\geq0$%
\begin{equation}
\operatorname*{Op}(q_{j}^{r}p_{j}^{s})=\operatorname*{Op}%
\nolimits_{\mathrm{BJ}}(q_{j}^{r}p_{j}^{s}). \label{opopbj}%
\end{equation}

\end{proposition}

\begin{proof}
It is sufficient to consider the case $n=1$; we write $q=q_{1}$ and $p=p_{1}$.
Taking the commutation formulas (\ref{commut1}) into account we can rewrite
the definitions (\ref{BJ1}) and (\ref{BJ2}) as%
\begin{equation}
\operatorname*{Op}\nolimits_{\mathrm{BJ}}(q^{r}p^{s})=\frac{1}{i\hbar
(r+1)(s+1)}[(\widehat{q})^{r+1},(\widehat{p})^{s+1}]. \label{fundamental}%
\end{equation}
We have%
\[
q^{r}p^{s}=\frac{1}{(r+1)(s+1)}\{q^{r+1},p^{s+1}\}
\]
and hence, using the axiom (GQ1),
\begin{equation}
\operatorname*{Op}(q^{r}p^{s})=\frac{1}{i\hbar(r+1)(s+1)}[(\widehat{q}%
)^{r+1},(\widehat{p})^{s+1}]; \label{comm}%
\end{equation}
the identity (\ref{opopbj}) follows using formula (\ref{fundamental}).
\end{proof}

Notice that formula (\ref{comm}) is interesting \emph{per se}: it shows that
the Born--Jordan quantization of a polynomial in the position and momentum
variables can be expressed as a linear combination of commutators.

\section{Quantization of $e^{\frac{i}{\hbar}(q_{0}q+p_{0}p)}$}

From now on we assume that $\widehat{q}_{j}$ and $\widehat{p}_{j}$ are the
usual operators \textquotedblleft multiplication by $q_{j}$\textquotedblright%
\ and $-i\hbar\partial_{x_{j}}$ (condition (CQ1). The result below is
essential because it is the key to the quantization of arbitrary observables.

\begin{lemma}
\label{prop2}Let $X(q_{0})=e^{\frac{i}{\hbar}q_{0}q}$ and $Y(p_{0}%
)=e^{\frac{i}{\hbar}p_{0}p}$. Let $\operatorname*{Op}$ be an arbitrary
quantization satisfying the axiom (CQ1). We have
\begin{equation}
\operatorname*{Op}(X(q_{0}))=e^{\frac{i}{\hbar}q_{0}\widehat{q}}\text{ \ and
}\operatorname*{Op}(Y(p_{0}))=e^{\frac{i}{\hbar}p_{0}\widehat{p}}
\label{tictac}%
\end{equation}
that is
\begin{equation}
\operatorname*{Op}(X(q_{0}))\psi(q)=e^{\frac{i}{\hbar}q_{0}q}\psi(q)\text{ \ ,
\ }\operatorname*{Op}(Y(p_{0}))\psi(q)=\psi(q+p_{0}). \label{ttxp}%
\end{equation}

\end{lemma}

\begin{proof}
It is sufficient to consider the case $n=1$, we write again $q=q_{1}$ and
$p=p_{1}$. Expanding the exponential $e^{\frac{i}{\hbar}q_{0}q}$ in a Taylor
series we have, in view of the continuity of $\operatorname*{Op}$ and using
the first equation (CQ1)%
\[
\operatorname*{Op}(X(q_{0}))\psi(q)=\sum_{k=0}^{\infty}\frac{1}{k!}\left(
\frac{i}{\hbar}q_{0}q\right)  ^{k}\psi(q)=e^{\frac{i}{\hbar}q_{0}q}\psi(q).
\]
Similarly, using the second equation (CQ1),
\[
\operatorname*{Op}(Y(p_{0}))\psi(q)=\sum_{k=0}^{\infty}\frac{1}{k!}\left(
\frac{i}{\hbar}p_{0}(-i\hbar\partial_{q})\right)  ^{k}\psi(q)=\psi(q+p_{0}).
\]

\end{proof}

Let us apply the result above to a quantization of Weyl's characteristic
function \cite{26} $M(q_{0},p_{0})=e^{\frac{i}{\hbar}(q_{0}q+p_{0}p)}$; using
the notation above $M(q_{0},p_{0})=X(q_{0})\otimes Y(p_{0})$. We will see that
$\widehat{M}(q_{0},p_{0})\neq\widehat{T}_{q}(q_{0})\otimes\widehat{Y}(p_{0})$.
In fact:

\begin{proposition}
\label{prop3}In what follows $\operatorname*{Op}$ is a quantization satisfying
the reduced Dirac condition (BJQ1). (i) The operator $\widehat{M}(q_{0}%
,p_{0})=\operatorname*{Op}(M(q_{0},p_{0}))$ is given by the formula%
\begin{equation}
\widehat{M}(q_{0},p_{0})=\operatorname{sinc}\left(  \tfrac{p_{0}q_{0}}{2\hbar
}\right)  e^{\frac{i}{\hbar}(q_{0}\widehat{q}+p_{0}\widehat{p})}
\label{formula}%
\end{equation}
where $\operatorname{sinc}t=\sin t/t$ if $t\neq0$, $\operatorname{sinc}0=1$.
(ii) We have $\widehat{M}(q_{0},p_{0})=0$ for all $(q_{0},p_{0})\in
\mathbb{R}^{2n}$ such that $p_{0}q_{0}\neq0$ and $p_{0}q_{0}\in2\pi
\hbar\mathbb{Z}$.
\end{proposition}

\begin{proof}
(i) If $q_{0}=0$ or $p_{0}=0$ the result is obvious. Assume $p_{0}q_{0}\neq0$.
The reduced Dirac rule (GQ1) yields%
\begin{align*}
\lbrack\widehat{X}(q_{0}),\widehat{Y}(p_{0})]  &  =i\hbar\operatorname*{Op}%
(\{e^{\frac{i}{\hbar}q_{0}q},e^{\frac{i}{\hbar}p_{0}p}\})\\
&  =\frac{1}{i\hbar}p_{0}q_{0}\operatorname*{Op}(e^{\frac{i}{\hbar}%
(q_{0}q+p_{0}p)})\\
&  =\frac{1}{i\hbar}p_{0}q_{0}\widehat{M}(q_{0},p_{0})
\end{align*}
that is%
\begin{align*}
\widehat{M}(q_{0},p_{0})  &  =\frac{i\hbar}{p_{0}q_{0}}(\widehat{X}%
(q_{0})\widehat{Y}(p_{0})-\widehat{Y}(p_{0})\widehat{X}(q_{0}))\\
&  =\frac{i\hbar}{p_{0}q_{0}}(e^{\frac{i}{\hbar}q_{0}\widehat{q}}e^{\frac
{i}{\hbar}p_{0}\widehat{p}}-e^{\frac{i}{\hbar}p_{0}\widehat{p}}e^{\frac
{i}{\hbar}q_{0}\widehat{q}}).
\end{align*}
In view of the Baker--Campbell--Hausdorff formula
\begin{equation}
e^{A+B}=e^{-\frac{1}{2}[A,B]}e^{A}e^{B}=e^{\frac{1}{2}[A,B]}e^{B}e^{A}
\label{BCH}%
\end{equation}
valid for all operators $A$ and $B$ commuting with $[A,B]$ we have%
\begin{align*}
e^{\frac{i}{\hbar}q_{0}\widehat{q}}e^{\frac{i}{\hbar}p_{0}\widehat{p}}  &
=e^{-\frac{1}{2i\hbar}p_{0}q_{0}}e^{\frac{i}{\hbar}(q_{0}\widehat{q}%
+p_{0}\widehat{p})}\\
e^{\frac{i}{\hbar}p_{0}\widehat{p}}e^{\frac{i}{\hbar}q_{0}\widehat{q}}  &
=e^{\frac{1}{2i\hbar}p_{0}q_{0}}e^{\frac{i}{\hbar}(q_{0}\widehat{q}%
+p_{0}\widehat{p})}%
\end{align*}
and hence%
\[
\widehat{M}(q_{0},p_{0})=\frac{i\hbar}{p_{0}q_{0}}(e^{-\frac{1}{2i\hbar}%
p_{0}q_{0}}-e^{\frac{1}{2i\hbar}p_{0}q_{0}})e^{\frac{i}{\hbar}(q_{0}%
\widehat{q}+p_{0}\widehat{p})}%
\]
which is formula (\ref{formula}). (ii) is obvious.
\end{proof}

\section{The Case of Arbitrary Observables}

Let $H$ be an element of $\mathcal{S}(\mathbb{R}^{2n})$; let $\mathcal{H}$ be
the Fourier transform of $H$, defined by%
\[
\mathcal{H}(q_{0},p_{0})=\left(  \tfrac{1}{2\pi\hbar}\right)  ^{n}\int
H(q,p)e^{-\frac{i}{\hbar}(q_{0}q+p_{0}p)}d^{n}pd^{n}q.
\]
in view of the Fourier inversion formula we have%
\[
H(q,p)=\left(  \tfrac{1}{2\pi\hbar}\right)  ^{n}\int\mathcal{H}(q_{0}%
,p_{0})e^{\frac{i}{\hbar}(q_{0}q+p_{0}p)}d^{n}p_{0}d^{n}q_{0}%
\]
Let $\operatorname*{Op}$ be any quantization; by continuity and linearity we
have%
\[
\operatorname*{Op}(H)=\left(  \tfrac{1}{2\pi\hbar}\right)  ^{n}\int
\mathcal{H}(q_{0},p_{0})\operatorname*{Op}(e^{\frac{i}{\hbar}(q_{0}%
(\cdot)+p_{0}(\cdot))})d^{n}p_{0}d^{n}q_{0}.
\]
Viewing the integral as a distribution bracket, this formula extends to
arbitrary $H\in\mathcal{S}^{\prime}(\mathbb{R}^{2n})$ by continuity, noting
that $\mathcal{S}(\mathbb{R}^{2n})$ is dense in $\mathcal{S}^{\prime
}(\mathbb{R}^{2n})$. This formula shows that every quantization is uniquely
determined by its action of the exponentials $e^{-\frac{i}{\hbar}(q_{0}%
q+p_{0}p)}$ (for a much more general context, see Bergeron and Gazeau
\cite{Gazeau}). For instance, if $\operatorname*{Op}(e^{\frac{i}{\hbar}%
(q_{0}(\cdot)+p_{0}(\cdot))})=e^{\frac{i}{\hbar}(q_{0}\widehat{q}%
+p_{0}\widehat{p})}$ we get the usual Weyl quantization of $H$ of the
observable $H$ \cite{Folland,Birk,75}. Suppose now that
\[
\operatorname*{Op}(e^{\frac{i}{\hbar}(q_{0}(\cdot)+p_{0}(\cdot))})=\widehat
{M}(q_{0},p_{0})
\]
where $\widehat{M}(q_{0},p_{0})$ is defined by formula (\ref{formula}). For
$H\in\mathcal{S}(\mathbb{R}^{2n})$ and $\psi\in\mathcal{S}(\mathbb{R}^{n})$ we
have%
\begin{equation}
\widehat{H}\psi(q)=\left(  \tfrac{1}{2\pi\hbar}\right)  ^{n}\int
\mathcal{H}(q_{0},p_{0})\widehat{M}(q_{0},p_{0})\psi(q)d^{n}p_{0}d^{n}q_{0}.
\label{BJ3}%
\end{equation}
Rewriting (\ref{BJ3}) as a distributional bracket%
\begin{equation}
\widehat{H}\psi=\left(  \tfrac{1}{2\pi\hbar}\right)  ^{n}\langle
\mathcal{H}(\cdot,\cdot),\widehat{M}(\cdot,\cdot)\psi\rangle\label{BJ4}%
\end{equation}
we can extend the definition of $\widehat{H}$ to arbitrary $H\in$
$\mathcal{S}^{\prime}(\mathbb{R}^{2n})$ noting that $\widehat{M}(q_{0}%
,p_{0})\psi\in\mathcal{S}(\mathbb{R}^{n})$.

Choosing $H=1$ we have $\mathcal{H}=(2\pi\hslash)^{n}\delta$ hence
$\langle\mathcal{H}(\cdot,\cdot),\widehat{M}(\cdot,\cdot)\psi\rangle
=(2\pi\hslash)^{n}\psi$ and $\operatorname*{Op}(1)=I_{\mathrm{d}}$ in view of
(\ref{BJ4}).

Part (ii) of Proposition \ref{prop3} shows that the correspondence
$H\longmapsto\widehat{H}$ is not injective: we have%
\[
\operatorname*{Op}(H+%
%TCIMACRO{\tsum \nolimits_{(q_{j},p_{j})\in\Lambda}}%
%BeginExpansion
{\textstyle\sum\nolimits_{(q_{j},p_{j})\in\Lambda}}
%EndExpansion
c_{j}e^{\frac{i}{\hbar}(q_{j}q+p_{j}p)})=\operatorname*{Op}(H)
\]
where $\Lambda$ is any finite lattice in $\mathbb{R}^{2n}$ consisting of
points $(q_{j},p_{j})$ such that $q_{j}p_{j}=2N\pi\hbar$ for an integer
$N\neq0$. The correspondence $H\longmapsto\widehat{H}$ is however surjective:
for every $\widehat{H}\in\mathcal{L(S}(\mathbb{R}^{n}),\mathcal{S}^{\prime
}(\mathbb{R}^{n}))$ there exists (a non-unique) $H\in\mathcal{S}^{\prime
}(\mathbb{R}^{2n})$ such that $\widehat{H}=\operatorname*{Op}(H)$. The proof
of this property is difficult and technical (it relies on the Paley--Wiener
theorem and the theory of division of distributions), and we refer to our
recent paper \cite{cogoni1} with Cordero and Nicola for a detailed treatment
of this issue.

There remains to show that Axiom (GQ3) (symmetry on a dense subspace) is
verified. In fact:

\begin{proposition}
If $H\in\mathcal{S}^{\prime}(\mathbb{R}^{2n})$ is a real distribution, then
$\langle\widehat{H}\psi,\phi\rangle=\langle\psi,\widehat{H}\phi\rangle$ for
all test functions $\phi,\psi\in\mathcal{S}(\mathbb{R}^{n})$.
\end{proposition}

\begin{proof}
Returning to integral notation for clarity, we begin by remarking that
(\ref{BJ4}) can be rewritten as
\begin{equation}
\widehat{H}\psi(q)=\left(  \tfrac{1}{2\pi\hbar}\right)  ^{n}\int
\mathcal{H}(q_{0},p_{0})\Theta(q_{0},p_{0})e^{\frac{i}{\hbar}(q_{0}\widehat
{q}+p_{0}\widehat{p})}\psi(q)d^{n}p_{0}d^{n}q_{0} \label{BJ5}%
\end{equation}
where the Cohen kernel \cite{26} $\Theta$ is given by
\[
\Theta(q_{0},p_{0})=\operatorname{sinc}\left(  \tfrac{p_{0}q_{0}}{2\hbar
}\right)  .
\]
Operators of the type (\ref{BJ5}) with arbitrary Cohen kernel $\Theta
\in\mathcal{S}^{\prime}(\mathbb{R}^{2n})$ are well-known in the literature and
one proves (\cite{26}, \S 4.7) that the formal adjoint of $\widehat{H}$ is
given by%
\begin{equation}
\widehat{H}^{\dag}\psi(q)=\left(  \tfrac{1}{2\pi\hbar}\right)  ^{n}%
\int\mathcal{H}^{\ast}(-q_{0},-p_{0})\Theta^{\ast}(-q_{0},-p_{0})e^{\frac
{i}{\hbar}(q_{0}\widehat{q}+p_{0}\widehat{p})}\psi(q)d^{n}p_{0}d^{n}q_{0}.
\end{equation}
In the present case we have $\Theta^{\ast}(-q_{0},-p_{0})=\Theta(q_{0},p_{0})$
hence $\widehat{H}^{\dag}=\widehat{H}$ requires that $\mathcal{H}^{\ast
}(-q_{0},-p_{0})=\mathcal{H}(q_{0},p_{0})$, which holds if and only if $H$ is real.
\end{proof}

\section{Discussion and Conclusion}

As follows from the Groenewold--van Hove obstruction the general Dirac
requirement
\begin{equation}
\lbrack\widehat{H},\widehat{K}]=i\hbar\operatorname*{Op}(\{H,K\})
\label{dirac}%
\end{equation}
is not compatible with a full-blown quantization; with some hindsight this can
be understood as follows: the notion of Poisson bracket is intimately related
to the symplectic structure underlying Hamiltonian mechanics (this is pretty
obvious when one works on a symplectic manifold $(M,\omega)$ since the Poisson
bracket is not defined \textit{ex nihilo}, but by contracting the symplectic
form $\omega$ with the Hamiltonian fields $X_{H}$ and $X_{K}$:
$\{H,K\}=i_{X_{K}}i_{X_{H}}\omega$. One could therefore say that, in a sense,
Dirac's condition (\ref{dirac}) tries very hard to force quantum mechanics to
mimic Hamiltonian mechanics by imposing symplectic covariance \cite{RMP}. Now,
it is reasonably well known (see \cite{TRANS} and the references therein) that
the only quantization enjoying such full symplectic covariance is the Weyl
correspondence \cite{26,Folland,Birk,75}. But the Weyl correspondence does not
satisfy the general Dirac condition (\ref{dirac}), as already follows from the
conflicting formulas (\ref{gh1}) and (\ref{gh2}). Also, our restriction of
(\ref{dirac}) to Hamiltonians of the type $H(q,p)=T(p)+V(q)$ shows why the
symplectic covariance properties of Born--Jordan quantization are limited to
linear symplectomorphisms of the type $(q,p)\longmapsto(p,-q)$ or
$(q,p)\longmapsto(L^{-1}q,(L^{-1})^{T}p)$: these are the only, symplectic
automorphisms $S$ (together with their products) for which $H\circ S$ is again
of the type above (see \cite{TRANS,RMP,Springer}).

Our results also makes clear that there can't be any canonical quantization
satisfying Dirac's condition (\ref{dirac}) in full generality, that is for all
functions $H$ and $K$: if such a quantization existed, then it would hold in
particular for $H,K\in C^{\infty}(X)\oplus C^{\infty}(X^{\ast})$. But then
this quantization is that of Born--Jordan, for which (\ref{dirac}) does not
hold for arbitrary $H$ and $K$.\ Notice that this argument actually gives a
new proof of the Groenewold--van Hove result.

A last remark: we have chosen to implement the Dirac correspondence rule
(\ref{dirac}) on a specific subspace of observables, those of the type
$T(p)+V(x);$ these do not form an algebra. It is not clear whether this space
of observables is a maximal one, nor is it clear whether one could recover
some other quantization schemes by changing this space of observables. We will
come back to these delicate questions in the future.

\begin{acknowledgement}
This work has been financed by the grant P27773 of the Austrian research
Foundation FWF. The author wants to express his gratitude to the Referee for
extremely valuable suggestions.
\end{acknowledgement}

\section*{APPENDIX}

Let us prove formula (\ref{commut1}). We begin by noting that the equalities%
\begin{align}
\lbrack(\widehat{q})^{r},(\widehat{p})^{s}]  &  =si\hbar\sum\limits_{j=0}%
^{r-1}(\widehat{q})^{r-1-j}(\widehat{p})^{s-1}(\widehat{q})^{j}\label{commuta}%
\\
\lbrack(\widehat{q})^{r},(\widehat{p})^{s}]  &  =ri\hbar\sum\limits_{j=0}%
^{s-1}(\widehat{p})^{s-1-j}(\widehat{q})^{r-1}(\widehat{p})^{j}
\label{commutb}%
\end{align}
are equivalent. In fact, swapping $\widehat{q}$ and $\widehat{p}$ in
(\ref{commuta}) amounts to changing the bracket $[\widehat{q},\widehat
{p}]=i\hbar$ into $[\widehat{p},\widehat{q}]=-i\hbar$ so that%
\[
\lbrack(\widehat{p})^{r},(\widehat{r})^{s}]=-si\hbar\sum\limits_{j=0}%
^{r-1}(\widehat{p})^{r-1-j}(\widehat{q})^{s-1}(\widehat{p})^{j};
\]
swapping $r$ and $s$ then yields (\ref{commutb}), taking into account the
antisymmetry of the commutator bracket. Let us prove (\ref{commutb}) by
induction on the integer $s\geq1$. Let $s=1$; then $[(\widehat{q}%
)^{r},\widehat{p}]=(\widehat{q})^{r}\widehat{p}-\widehat{p}(\widehat{q})^{r}$
and we have, by repeated use of $[\widehat{q},\widehat{p}]=i\hbar$
\begin{align*}
\widehat{p}(\widehat{q})^{r}  &  =\widehat{p}\widehat{q}(\widehat{q}%
)^{r-1}=\widehat{q}\widehat{p}(\widehat{q})^{r-1}-i\hbar(\widehat{q})^{r-1}\\
&  =(\widehat{q})^{r}\widehat{p}-ri\hbar(\widehat{q})^{r-1}%
\end{align*}
that is $[(\widehat{q})^{r},\widehat{p}]=ri\hbar(\widehat{q})^{r-1}$ which
proves (\ref{commutb}) in this case. Let now $s$ be an arbitrary integer
$\geq2$ and assume that%
\begin{equation}
\lbrack(\widehat{q})^{r},(\widehat{p})^{s-1}]=ri\hbar\sum\limits_{j=0}%
^{s-2}(\widehat{p})^{s-2-j}(\widehat{q})^{r-1}(\widehat{p})^{j}.
\label{assumption}%
\end{equation}
We then have%
\begin{align*}
\lbrack(\widehat{q})^{r},(\widehat{p})^{s}]  &  =(\widehat{q})^{r}(\widehat
{p})^{s}-(\widehat{p})^{s}(\widehat{q})^{r}\\
&  =(\widehat{q})^{r}(\widehat{p})^{s-1}\widehat{p}-(\widehat{p}%
)^{s-1}\widehat{p}(\widehat{q})^{r}\\
&  =(\widehat{q})^{r}(\widehat{p})^{s-1}\widehat{p}-(\widehat{p}%
)^{s-1}((\widehat{q})^{r}\widehat{p}-ri\hbar(\widehat{q})^{r-1})\\
&  =[(\widehat{q})^{r},(\widehat{p})^{s-1}]\widehat{p}+ri\hbar(\widehat
{p})^{s-1}(\widehat{q})^{r-1}.
\end{align*}
In view of assumption (\ref{assumption}) this is%
\begin{align*}
\lbrack(\widehat{q})^{r},(\widehat{p})^{s}]  &  =ri\hbar\sum\limits_{j=0}%
^{s-2}(\widehat{p})^{s-2-j}(\widehat{q})^{r-1}(\widehat{p})^{j+1}%
+ri\hbar(\widehat{p})^{s-1}(\widehat{q})^{r-1}\\
&  =ri\hbar\sum\limits_{j=0}^{s-1}(\widehat{p})^{s-1-j}(\widehat{q}%
)^{r-1}(\widehat{p})^{j}%
\end{align*}
which completes the proof.


\begin{thebibliography}{99}                                                                                               %


\bibitem {tareq}T. S. Ali and M. Engli\v{s}. Quantization methods: a guide for
physicists and analysts. Rev. Math. Phys. 17(04) (2005) 391--490.

\bibitem {gazeau}H. Bergeron and J. P Gazeau. Integral quantizations with two
basic examples. Annals of Physics 344 (2014) 43--68.

\bibitem {berndt}R. Berndt. An introduction to symplectic geometry.
Providence, Rhode Island: American Mathematical Society, 2001.

\bibitem {13}P. Boggiatto, G. De Donno, A. Oliaro, Time-Frequency
Representations of Wigner Type and Pseudo-Differential Operators, Trans. Amer.
Math. Soc., 362(9) 4955--4981 (2010).

\bibitem {16}M. Born, P. Jordan. Zur Quantenmechanik, Zeits. Physik 34,
858--888 (1925).

\bibitem {20}L. Castellani, Quantization Rules and Dirac's Correspondence, Il
Nuovo Cimento 48A(3), (1978) 359--368.

\bibitem {chernoff}P. R. Chernoff. Mathematical obstructions to quantization.
Hadronic J. 4 (1981) 879--898.

\bibitem {26}L. Cohen. The Weyl operator and its generalization. Springer
Science \& Business Media, 2012.

\bibitem {cogoni1}E. Cordero, M. de Gosson, and F. Nicola. On the
invertibility of Born--Jordan quantization. Journal de Math\'{e}matiques Pures
et Appliqu\'{e}es (2015).

\bibitem {Dirac}P. A. M. Dirac, Principles of Quantum Mechanics. USA: Oxford
University Press, 1982.

\bibitem {34}H. B. Domingo, E. A. Galapon. Generalized Weyl transform for
operator ordering:\ Polynomial functions in phase space. J.\ Math. Phys. 56,
022104 (2015).

\bibitem {Folland}G. B. Folland. Harmonic Analysis in Phase space, Annals of
Mathematics studies, Princeton University Press, Princeton, N.J. 1989.

\bibitem {Birk}M. de Gosson, Symplectic Geometry and Quantum Mechanics.
Birkh\"{a}user, Basel, series \textquotedblleft Operator Theory: Advances and
Applications\textquotedblright\ (subseries: \textquotedblleft Advances in
Partial Differential Equations\textquotedblright), Vol. 166, (2006).

\bibitem {52}M. de Gosson. On the usefulness of an index due to Leray for
studying the intersections of Lagrangian and symplectic paths. J. Math. Pures
Appl. 91, 598--613 (2009).

\bibitem {TRANS}M. de Gosson. Symplectic covariance properties for Shubin and
Born--Jordan pseudo-differential operators. Trans. Amer. Math. Soc. 365(6)
(2013), 3287--3307.

\bibitem {RMP}M. de Gosson. Paths of canonical transformations and their
quantization. Rev. Math. Phys. 27(6) (2015) 1530003.

\bibitem {physrep}M. de Gosson. From Weyl to Born--Jordan quantization: the
Schr\"{o}dinger representation revisited. Phys. Rep. 623 (2016) 1--58.

\bibitem {Springer}M. de Gosson. Born-Jordan Quantization: Theory and
Applications. Springer 2016.

\bibitem {gotay}M. J. Gotay. On the Groenewold--Van Hove problem for
$\mathbb{R}^{2n}$ J. Math. Phys. 40(4) (1999) 2107--2116.

\bibitem {gotuy}M. J. Gotay, H. B. Grundling, and G. M. Tuynman. Obstruction
results in quantization theory. J. Nonlinear Sci. 6(5) (1996) 469--498.

\bibitem {groenewold}H. J. Groenewold. On the principles of elementary quantum
mechanics. Physica 12 (1946) 405--460.

\bibitem {GS}V. Guillemin and S. Sternberg. Symplectic techniques in physics.
Cambridge University Press, 1990.

\bibitem {Kauffmann}S. K. Kauffmann. Unambiguous quantization from the maximum
classical correspondence that is self-consistent: the slightly stronger
canonical commutation rule Dirac missed. Found. Phys. 41(5) (2011) 805--819.

\bibitem {75}R. G. Littlejohn. The semiclassical evolution of wave packets,
Phys. Rep. 138, 4-5 193--291 (1986).

\bibitem {JCPain}J.-C. Pain. Commutation relations of operator monomials.
Journal of Physics A: Mathematical and Theoretical 46(3) (2012) 035304.

\bibitem {vanhove1}L. van Hove. Sur certaines repr\'{e}sentations unitaires
d'un group infini de transformations. Proc. Roy. Acad. Sci. Belgium 26, 1--102.

\bibitem {vanhove2}L. van Hove. Sur le probl\`{e}me des relations entre les
transformations unitaires de la m\'{e}canique quantique et les transformations
canoniques de la m\'{e}canique classique. Acad. Roy. Belgique Bull. Cl. Sci.
(5) 37 (1951) 610--620.

\bibitem {101}M. A. Shubin. Pseudodifferential Operators and Spectral Theory,
Springer-Verlag, 1987 [original Russian edition in Nauka, Moskva 1978].

\bibitem {wilcox}R. M. Wilcox. Exponential operators and parameter
differentiation in quantum physics. Journal of Mathematical Physics 8(4)
(1967) 962--982.
\end{thebibliography}
\end{document}